\documentclass[12pt]{amsart}

\usepackage{amsfonts}
\usepackage{amsmath}

\newtheorem{thm}{Theorem}[section]

\newtheorem{cor}[thm]{Corollary}
\newtheorem{lem}[thm]{Lemma}

\newtheorem*{prob*}{Problem}
\newtheorem*{thm*}{Theorem}

\theoremstyle{definition}

\newtheorem*{defn*}{Definition}

\newtheorem*{rem*}{Remark}
\numberwithin{equation}{section}

\newcommand{\C}{\mathbb C}
\newcommand{\R}{\mathbb R}

\DeclareMathOperator{\const}{const}
\DeclareMathOperator{\supp}{supp}

\newcommand{\im}{\mathop{\mathrm{Im}}}

\begin{document}
\title[Universality theorem for  characteristic polynomials]
 {\bf{A universality theorem for ratios of random characteristic polynomials}}

\author{Jonathan Breuer and Eugene Strahov}

\thanks{Einstein Institute of Mathematics, The Hebrew University of
Jerusalem, Givat Ram, Jerusalem 91904. E-mail: jbreuer@math.huji.ac.il,
strahov@math.huji.ac.il. The first author (J.~B.) is supported in part by the US-Israel Binational Science Foundation (BSF) Grant No.\ 2010348 and by the Israel Science foundation (ISF) Grant No.\ 1105/10. The second author (E.~S.) is supported in part by the US-Israel Binational Science
Foundation (BSF) Grant No.\ 2006333,
 and by the Israel Science Foundation (ISF) Grant No.\ 1441/08.\\ }


\commby{}
\begin{abstract}
We consider asymptotics of ratios of random characteristic polynomials associated
with orthogonal polynomial ensembles. Under some natural conditions on the measure in the definition of the orthogonal polynomial ensemble
we  establish a universality limit for these ratios. 
\end{abstract}
\maketitle
\section{Introduction}
\subsection{Formulation of the problem}\label{OrthogonalPolynomialEnsembleSection}

In this article we consider orthogonal polynomial ensembles of $n$ particles on $\R$.
Such ensembles are described by a positive Borel measure $\mu$ with finite moments, and the associated
 distribution function for the particles $\{x_1,\ldots,x_n\}$ of the form
$$
d\mathbb{P}_{\mu,n}(x)=\frac{1}{Z_n}\triangle_n(x)^2\prod\limits_{i=1}^nd\mu(x_i).
$$
Here $Z_n$ is the normalization constant,
$$
Z_n=\int\ldots\int\triangle_n(x)^2\prod\limits_{i=1}^nd\mu(x_i),
$$
and
$$
\triangle_n(x)=\prod\limits_{n\geq i>j\geq 1}(x_i-x_j)
$$
is the Vandermonde determinant. For symmetric functions $f(x)=f(x_1,\ldots,x_n)$ of the $x_i's$,
$$
\left\langle f(x)\right\rangle_{\mu}\equiv \frac{1}{Z_n}\int\ldots\int f(x)\triangle_n(x)^2\prod\limits_{i=1}^nd\mu(x_i)
$$
denotes the average of $f$ with respect to $d\mathbb{P}_{\mu,n}(x)$.

Given a set of $n$ points, $\{x_1,\ldots,x_n\}$, let
$$
D_n^{\{x_1, \ldots,x_n\}}(\alpha)=\prod\limits_{i=1}^n(\alpha-x_i),
$$
considered as a polynomial in $\alpha \in \mathbb{C}$. When the $x_i$'s are chosen to be the \emph{random} particle locations for an orthogonal polynomial ensemble, this becomes a random polynomial which we denote simply by $D_n(\alpha)$. By extension from the random matrix case (see below), $D_n(\alpha)$ is known as the \emph{characteristic polynomial} for the orthogonal polynomial ensemble.  


The goal of the present paper is to study the large $n$ asymptotics of the averages
\begin{equation}\label{averages}
\left\langle\frac{D_n(\alpha_1)\ldots D_n(\alpha_k)}{D_n(\beta_1)\ldots D_n(\beta_k)}\right\rangle_{\mu}.
\end{equation}
In particular, we focus on the bulk of the support for measures $\mu$, and aim to establish universality of the scaling limit of averages (\ref{averages}) under mild assumptions on $\mu$. Our main result says roughly that for $\mu$ locally absolutely continuous with a bounded Radon-Nikodym derivative, universality for the reproducing kernel implies universality for \eqref{averages}. It should be emphasized that aside for the existence of moments, we make no global assumptions on $\mu$. In particular, we do not assume that $\mu$ is globally absolutely continuous or that $\supp(\mu)$ is compact.   


\subsection{Motivation and remarks on related works}
In the case when $d\mu(x)=e^{-V(x)}dx$ one can interpret the particles in the definition of the orthogonal polynomial ensembles as eigenvalues of a random Hermitian matrix
taken from a Unitary Ensemble of Random Matrix Theory (RMT). The most well studied case is that of $V(x)=x^2$, called the Gaussian Unitary Ensemble. In this context random characteristic polynomials are indeed the characteristic polynomials of a random Hermitian matrix. 

Averages of characteristic polynomials of random matrices are basic objects of interest in RMT, and were first considered by Andreev and Simons \cite{andreev}, Brezin and Hikami \cite{BrezinHikami} and Keating and Snaith \cite{KeatingSnaith}.  In particular, such averages are used to make predictions about zeros of the Riemann-zeta function on the critical line
(see, for example, Keating and Snaith \cite{KeatingSnaith},  Conrey, Farmer,  Keating,  Rubinstein and  Snaith \cite{Conrey}, the  survey article by Keating and Snaith \cite{KeatingSnaith1}, and references therein). Averages \eqref{averages} are related to certain important distribution functions studied in physics of quantum chaotic systems. Two examples involve the curvature distribution of energy levels of a chaotic system and the statistics of the local Green functions, in particular, the joint distribution of local
denity of states; see Andreev and Simons \cite{andreev} and references therein.  Many other uses are described, for example, in Brezin \cite{BrezinSurvey}.

For averages \eqref{averages} a number of algebraic and asymptotic results is available in the literature. Papers by Baik, Deift, and Strahov \cite{BaikDeiftStrahov}, Fyodorov and Strahov \cite{FyodorovStrahov} and 
Borodin and Strahov \cite{BorodinStrahov} give explicit determinantal representations for \eqref{averages}.
These representations can be used for the asymptotic analysis as $n\rightarrow\infty$. In particular,
the asymptotics of \eqref{averages} were investigated in Strahov and Fyodorov \cite{FyodorovStrahov} in the case when $d\mu(x)=e^{-V(x)}dx$, and $V(x)$ is an even polynomial. Strahov and Fyodorov \cite{FyodorovStrahov}
deal with the asymptotic in the bulk of the spectrum.  Vanlessen \cite{Vanlessen} shows that  for certain class of unitary ensembles of
Hermitian matrices, averages \eqref{averages} have universal asymptotic  behavior at the origin of the spectrum. The asymptotic analysis in \cite{FyodorovStrahov}, and in \cite{Vanlessen}
is based on the reformulation of an orthogonal polynomial problem as a Riemann-Hilbert problem by Fokas, Its and Kitaev \cite{fokas}. The Riemann-Hilbert problem is then analyzed asymptotically using the noncommutative steepest-descent
method introduced by Deift and Zhou, see Deift \cite{deift} and references therein.

In recent years, it has become evident that orthogonal polynomial ensembles play a role in probabilistic models other than RMT as well (see, for example, the survey paper by  K\"onig \cite{Konig}). As the relevant measures in these models are not necessarily of the form $e^{-V(x)}dx$, it is of interest to study the problem of universality limits for basic quantities of interest for a broader class of measures. Lubinsky's universality theorems regarding bulk universality for the reproducing kernel (see Lubinsky \cite[Theorem 1.1]{Lubinsky1}, and \cite[Theorem 1.1]{Lubinsky3}, and also Findley \cite{findley}, Simon \cite{Simon-ext}, Totik \cite{totik} and Avila-Last-Simon \cite{als} for extensions of Lubinsky's results and methods) are important steps in this direction. Our goal in this paper is to establish the corresponding universality limits for averages \eqref{averages} in the bulk of the support of $\mu$.



\subsection{Description of the main result}
For $n=0,1,2,\ldots $ we introduce the orthonormal polynomials associated with $\mu$,
$$
p_n(x)=\gamma_nx^n+\ldots ,\;\;\gamma_n>0.
$$
The orthonormality conditions are
$$
\int p_j(x)p_k(x)d\mu(x)=\delta_{jk}.
$$

The $n$th reproducing kernel (also known as the Christoffel-Darboux kernel) for $\mu$ is
$$
K_n(x,y)=\sum\limits_{k=0}^{n-1}p_k(x)p_k(y).
$$
If $d \mu(t)=w(t)dt$ in a neighborhood of $x$, we also define the normalized kernel to be
$$
\widetilde{K}_n(x,y)=w(x) K_n(x,y).
$$

The Christoffel-Darboux formula enables one to rewrite $K_n(x,y)$ as
\begin{equation}\label{ChristoffelDarboux}
K_n(x,y)=\frac{\gamma_{n-1}}{\gamma_{n}}\frac{p_n(x)p_{n-1}(y)-p_{n-1}(x)p_n(y)}{x-y}.
\end{equation}
The reproducing kernel plays a special role in the theory of orthogonal polynomial ensembles. This is because
the correlation functions for the orthogonal polynomial ensemble can be expressed as determinants
of a matrix whose entries are given by values of the reproducing kernel (see, for example, Deift \cite{deift}). 

We say $K_n$ has a universal limit at $x$, if for any $a,b \in \mathbb{C}$ we have
\begin{equation} \label{universal}
\underset{n\rightarrow\infty}{\lim}\frac{K_n(x+\frac{a}{\widetilde{K}_n(x,x)},x+\frac{b}{\widetilde{K}_n(x,x)})}{K_n(x,x)}=\frac{\sin\pi(a-b)}{\pi(a-b)}=\mathbb{S}(a,b).
\end{equation}
We say $K_n$ has a \emph{uniform universal limit} at $x$, if the limit in \eqref{universal} is uniform for $a, b$ in compact subsets of $\mathbb{C}$.
 
Two approaches introduced recently by Lubinsky (\cite{Lubinsky1}-\cite{Lubinsky3}) make it possible to establish universality for $K_n$ under relatively mild conditions on $\mu$. These approaches were further extended and generalized by Findley \cite{findley}, Simon \cite{Simon-ext}, Totik \cite{totik} and Avila-Last-Simon \cite{als}. A typical result is: 

\begin{thm}\label{LubinskyTheorem}
Let $\mu$ be a probability measure on $\mathbb{R}$ with compact support that is regular in the sense of Stahl and Totik \cite{stahl-totik}. Suppose $x\in \supp(\mu)$ has a neighborhood, $J$, such that $\mu$ is absolutely continuous in $J$: $d\mu (t)=w(t)dt$ for $t \in J$. Assume further, that $w$ is positive and continuous at $x$. Then uniformly for $a,b$ in compact subsets of the complex plane, we have
\begin{equation}
\underset{n\rightarrow\infty}{\lim}\frac{K_n(x+\frac{a}{\widetilde{K}_n(x,x)},x+\frac{b}{\widetilde{K}_n(x,x)})}{K_n(x,x)}=\mathbb{S}(a,b).
\end{equation}
In other words, $K_n$ has a uniform universal limit at $x$.
\end{thm}
{\bf Remarks.} \\ 
(a) Theorem \ref{LubinskyTheorem} implies universality in the bulk for the $m$-point correlation function for the orthogonal polynomial ensemble (see Lubinsky \cite[Section 1]{Lubinsky2}). \\
(b) The Theorem follows, for example, from \cite[Theorem 1.1]{Lubinsky3} combined with remark (e) there. The result was obtained previously, however, by both Simon \cite{Simon-ext} and Totik \cite{totik}, using a modification of Lubinsky's approach from \cite{Lubinsky1}. \\
(c) The requirement of continuity can be relaxed to a Lebesgue  point type condition, assuming boundedness and uniform positivity of $w$ in a neighborhood of $x$. This is Theorem 1.2 of \cite{Lubinsky3}. \\
(d) Avila, Last and Simon \cite{als}, using a modification of the approach of \cite{Lubinsky3}, obtain universality for certain measures whose support is a positive Lebesgue measure Cantor set. \\
(e) While it seems that local absolute continuity of $\mu$ is almost sufficient for universality at $x$, it is certainly not necessary: Breuer \cite{breuer} has recently shown there are purely singular measures such that universality holds uniformly for $x$ in an interval (and $a,b \in \mathbb{R}$). Of course, $\widetilde{K}_n(x,x)$ cannot be defined for purely singular measures, so the statement needs to be slightly modified (see \cite{breuer} for details).
\smallskip

In this paper we prove an analogue of Theorem \ref{LubinskyTheorem} for averages \eqref{averages}. Here is our main result.
\begin{thm} \label{MainTheorem}
Let $\mu$ be a probabilty measure on $\R$ with finite moments. Let $x \in \supp(\mu)$ be such that: \\
$(i)$ There exists an interval, $J$, with $x \in J$ such that $\mu$ is absolutely continuous in $J$: $d\mu(t)=w(t)dt$ for $t \in J$. Moreover, $w \in L^\infty(J,dt)$. \\
$(ii)$ $w(x)>0$ and $x$ is a Lebesgue point of $w$, by which we mean that 
$$
\lim_{r \rightarrow 0+} \frac{1}{r}\int_{x-r}^{x+r}|w(t)-w(x)|dt =0.
$$ \\
$(iii)$ $K_n$ has a uniform universal limit at $x$. \\

Under these assumptions, for any pairwise distinct $\alpha_i$'s and $\beta_j$'s, such that $\alpha_1,\ldots,\alpha_k\in\R$, $\beta_1,\ldots,\beta_k\in\C\setminus\R$
\begin{equation}
\underset{n\rightarrow\infty}{\lim}\left\langle\prod\limits_{j=1}^k\frac{D_n\left(x+\frac{\alpha_j}{\widetilde{K}_n(x,x)}\right)}{D_n\left(x+\frac{\beta_j}{\widetilde{K}_n(x,x)}\right)}\right\rangle_{\mu}
=(-1)^{\frac{k(k+1)}{2}}
\frac{\triangle(\beta,\alpha)}{\triangle(\beta)^2\triangle(\alpha)^2}\det\left(\mathbb{W}(\beta_i,\alpha_j)\right)_{i,j=1}^k,
\nonumber
\end{equation}
where
$$
\mathbb{W}(\beta,\alpha)=\frac{1}{\beta-\alpha}+\int\limits_{-\infty}^{+\infty}\frac{\sin(\pi(s-\alpha))ds}{\pi(s-\alpha)(s-\beta)}.
$$
\end{thm}

As an immediate corollary of Theorems \ref{LubinskyTheorem} and \ref{MainTheorem}, we obtain
\begin{cor} \label{corollary}
Let $\mu$ be a probability measure on $\mathbb{R}$ with compact support that is regular in the sense of Stahl and Totik. Suppose $x\in \supp(\mu)$ has a neighborhood, $J$, such that $\mu$ is absolutely continuous in $J$: $d\mu (t)=w(t)dt$ for $t \in J$, with $w$ bounded in $J$. Assume further, that $w$ is positive and continuous at $x$. Then for any pairwise distinct $\alpha_i$'s and $\beta_j$'s, such that $\alpha_1,\ldots,\alpha_k\in\R$, $\beta_1,\ldots,\beta_k\in\C\setminus\R$
\begin{equation}
\underset{n\rightarrow\infty}{\lim}\left\langle\prod\limits_{j=1}^k\frac{D_n\left(x+\frac{\alpha_j}{\widetilde{K}_n(x,x)}\right)}{D_n\left(x+\frac{\beta_j}{\widetilde{K}_n(x,x)}\right)}\right\rangle_{\mu}
=(-1)^{\frac{k(k+1)}{2}}
\frac{\triangle(\beta,\alpha)}{\triangle(\beta)^2\triangle(\alpha)^2}\det\left(\mathbb{W}(\beta_i,\alpha_j)\right)_{i,j=1}^k,
\nonumber
\end{equation}
\end{cor}

{\bf Remarks.} \\
(a) The function $\mathbb{W}(\beta,\alpha)$ can be rewritten as
$$
\mathbb{W}(\beta,\alpha)=\left\{
                           \begin{array}{ll}
                             \frac{e^{i\pi (\beta-\alpha)}}{\beta-\alpha}, & \im{\beta}>0, \\
                            \frac{e^{-i\pi (\beta-\alpha)}}{\beta-\alpha}, & \im{\beta}<0.
                           \end{array}
                         \right.
$$
The formula above can be obtained  from that of Theorem \ref{MainTheorem} by contour integration.\\
(b) If $d\mu(y)=\const e^{-V(y)}dy$, and $V(y)$ is an even polynomial, then Theorem \ref{MainTheorem} reduces to the result obtained in Strahov and Fyodorov \cite[Section 2.4]{StrahovFyodorov}.\\
(c) Note that the assumptions of Theorem \ref{MainTheorem} do not include compact support of $\mu$. The theorem essentially says that for any measure with finite moments, local absolute continuity combined with a uniform universal limit for $K_n$ at $x$ imply uniform limits for averages \eqref{averages} at $x$. \\
(d) As explained in remark (e) after Theorem \ref{LubinskyTheorem}, there are purely singular measures for which $K_n$ has universal limits. As the absolute continuity of $\mu$ around $x$ is essential to the proof of Theorem \ref{MainTheorem} we, unfortunately, have nothing to say about this case.  

The rest of the paper is devoted to the proof of Theorem \ref{MainTheorem}. Namely, in Section \ref{SectionAlgebra} we present a useful algebraic formula
for averages \eqref{averages} (Theorem \ref{AlgebraicTheorem}).  This formula reduces the investigation of asymptotics of \eqref{averages} to that of
the Cauchy transform of the reproducing kernel,
$$
\int\frac{K_n(t,\alpha)d\mu(t)}{t-\beta}.
$$
Theorem \ref{AsymptoticsKernel} establishes  universality for the Cauchy transform of the reproducing kernel.
Theorem \ref{MainTheorem} (which is the main result of this work) is then a simple corollary of Theorem \ref{AlgebraicTheorem}, and of Theorem \ref{AsymptoticsKernel}.

\section{A formula for ratios of characteristic polynomials}\label{SectionAlgebra}

\begin{thm}\label{AlgebraicTheorem}
Let $1\leq k\leq n$, and assume that $\alpha_1,\ldots,\alpha_k,\beta_1,\ldots,\beta_k$ are pairwise distinct complex numbers.
Moreover, assume that $\im(\beta_j)\neq 0$ for $j=1,\ldots,k$. Then we have
$$
\left\langle\prod\limits_{j=1}^k\frac{D_n(\alpha_j)}{D_n(\beta_j)}\right\rangle_{\mu}=(-1)^{\frac{k(k+1)}{2}}
\frac{\triangle(\beta,\alpha)}{\triangle(\beta)^2\triangle(\alpha)^2}\det\left(W_n(\beta_i,\alpha_j)\right)_{i,j=1}^k,
$$
where the two-point function $W_n(\beta,\alpha)$ is defined by
$$
W_n(\beta,\alpha)=\frac{1}{\beta-\alpha}+\int\frac{K_n(t,\alpha)d\mu(t)}{t-\beta}.
$$
\end{thm}
\begin{proof}A formula of two-point function type for the average of ratios of characteristic polynomials was obtained by Baik, Deift and Strahov \cite[Section III]{BaikDeiftStrahov}. Theorem 3.3 in \cite{BaikDeiftStrahov} gives
$$
\left\langle\prod\limits_{j=1}^k\frac{D_n(\alpha_j)}{D_n(\beta_j)}\right\rangle_{\mu}=(-1)^{\frac{k(k-1)}{2}}\widetilde{\gamma}_{n-1}^k
\frac{\triangle(\beta,\alpha)}{\triangle(\beta)^2\triangle(\alpha)^2}\det\left(\widetilde{W}_n(\beta_i,\alpha_j)\right)_{i,j=1}^k,
$$
where
$$
\widetilde{W}_n(\beta,\alpha)=\frac{\widetilde{h}_n(\beta)\pi_{n-1}(\alpha)-\widetilde{h}_{n-1}(\beta)\pi_n(\alpha)}{\beta-\alpha}.
$$
In the formulae above $\pi_l(\alpha)$ is the $l$th monic orthogonal polynomial associated with $\mu$,
the function $\widetilde{h}_l(\beta)$ is the Cauchy  transform of $\pi_l(\alpha)$,
$$
\widetilde{h}_l(\beta)=\frac{1}{2\pi i}\int\frac{\pi_l(t)d\mu(t)}{t-\beta},
$$
and
$$
\widetilde{\gamma}_{n-1}=-2\pi i\gamma_{n-1}^2.
$$

Clearly, the formula for the average of ratios of characteristic polynomials can be rewritten as
$$
\left\langle\prod\limits_{j=1}^k\frac{D_n(\alpha_j)}{D_n(\beta_j)}\right\rangle_{\mu}=(-1)^{\frac{k(k+1)}{2}}
\frac{\triangle(\beta,\alpha)}{\triangle(\beta)^2\triangle(\alpha)^2}\det\left(W_n(\beta_i,\alpha_j)\right)_{i,j=1}^k,
$$
where
$$
W_n(\beta,\alpha)=\frac{\gamma_{n-1}}{\gamma_n}\frac{h_n(\beta)p_{n-1}(\alpha)-h_{n-1}(\beta)p_n(\alpha)}{\beta-\alpha},
$$
and where
$$
h_l(\beta)=\int\frac{p_l(t)d\mu(t)}{t-\beta}.
$$
To obtain the formula for $W_n(\beta,\alpha)$ stated in the Theorem we note that
\begin{equation}
\begin{split}
&\frac{\gamma_{n-1}}{\gamma_n}\frac{h_n(\beta)p_{n-1}(\alpha)-h_{n-1}(\beta)p_n(\alpha)}{\beta-\alpha}\\
&=\frac{\gamma_{n-1}}{\gamma_n}\frac{1}{\beta-\alpha}\left(\left(\int\frac{p_n(t)d\mu(t)}{t-\beta}\right)p_{n-1}(\alpha)-
\left(\int\frac{p_{n-1}(t)d\mu(t)}{t-\beta}\right)p_{n}(\alpha)\right)\\
&=\frac{1}{\beta-\alpha}\int\frac{t-\alpha}{t-\beta}K_n(t,\alpha)d\mu(t)\\
&=\frac{1}{\beta-\alpha}+\int\frac{K_n(t,\alpha)d\mu(t)}{t-\beta},
\end{split}
\nonumber
\end{equation}
where in the second equality we have used formula  (\ref{ChristoffelDarboux}), and in the third equality we have used the reproducing property
of $K_n$.
\end{proof}

\section{Universality for the Cauchy transform of the reproducing kernel}
\begin{thm}\label{AsymptoticsKernel}
Let $x \in \supp(\mu)$ be such that conditions $(i)-(iii)$ of Theorem \ref{MainTheorem} are satisfied. Then for any $\alpha\in\R$, $\beta\in\C\setminus\R$
$$
\lim_{n \rightarrow \infty}  \frac{1}{\widetilde{K}_n(x,x)} \int \frac{ K_n \left(x+\frac{\alpha}{\widetilde{K}_n(x,x)},t \right)}{t-x-\frac{\beta}{\widetilde{K}_n(x,x)}}d\mu(t)=\int\limits_{-\infty}^{+\infty} \frac{\mathbb{S}(\alpha,s)}{s-\beta}ds,
$$
where
$$
\mathbb{S}(\alpha,s)=\frac{\sin\pi(\alpha-s)}{\pi(\alpha-s)}.
$$
\end{thm}
\begin{rem*}
In case of the Chebyshev weight for example,
$$
d\mu(y)=w(y)dy,\;\; w(y)=\frac{1}{\sqrt{1-y^2}},\;\; y\in(-1,1),
$$
Theorem \ref{AsymptoticsKernel} can be checked by direct computations. These computations are similar to those of  \cite[Section 2]{Lubinsky2}, for the reproducing kernel associated with the Chebyshev weight.
\end{rem*}
Theorem \ref{AsymptoticsKernel} follows from the following
\begin{lem}\label{1}
Fix $\alpha \in \R$ and $\beta \in \C \setminus \R$.
Under the conditions of Theorem \ref{MainTheorem}, for any $M \geq 2|\beta|$,
\begin{equation} \label{0.2}
\begin{split}
&\limsup_{n\rightarrow \infty} \left |  \frac{1}{\widetilde{K}_n(x,x)} \int \frac{ K_n \left(x+\frac{\alpha}{\widetilde{K}_n(x,x)},t \right)}{t-x-\frac{\beta}{\widetilde{K}_n(x,x)}}d\mu(t)-\int \frac{\mathbb{S}(\alpha,s)}{s-\beta}ds \right| \\
& \leq \left| \int_{\R \setminus [-M,M] } \frac{\mathbb{S}(\alpha,s)}{s-\beta}ds \right|+8 \sqrt{\frac{\| w \|_J}{w(0) M}}
\end{split}
\end{equation}
where $\|w\|_J$ is the essential supremum of $w$ on $J$.
\end{lem}

\begin{proof}[Proof of Theorem \ref{AsymptoticsKernel} assuming Lemma \ref{1}]
Fix $\alpha \in \R$ and $\beta \in \C \setminus \R$. As $M \rightarrow \infty$ the righthand side of inequality (\ref{0.2}) goes to zero. This is because \
$$
\frac{\mathbb{S}(\alpha,s)}{s-\beta}=\frac{\sin(\pi(s-\alpha))}{\pi(s-\alpha)(s-\beta)}
$$
is integrable so its tail goes to zero.
Since $M$ is arbitrary, we get  the limit relation in Theorem \ref{AsymptoticsKernel}.
\end{proof}
\begin{proof}[Proof of Lemma \ref{1}]
By shifting the measure, we can assume that $x=0$, which we henceforth do for ease of notation. Fix $\alpha \in \R$ and $\beta \in \C \setminus \R$, and fix $M \geq 2|\beta|$. Write
\begin{equation}
\begin{split}
&\left |  \frac{1}{\widetilde{K}_n(0,0)} \int \frac{ K_n \left(\frac{\alpha}{\widetilde{K}_n(0,0)},t\right)}{t-\frac{\beta}{\widetilde{K}_n(0,0)}}d\mu(t)-\int \frac{\mathbb{S}(\alpha,s)}{s-\beta}ds \right|  \\
& \leq \left |  \frac{1}{\widetilde{K}_n(0,0)} \int \frac{ K_n \left(\frac{\alpha}{\widetilde{K}_n(0,0)},t \right)}{t-\frac{\beta}{\widetilde{K}_n(0,0)}}d\mu(t)-\int_{-M}^M \frac{\mathbb{S}(\alpha,s)}{s-\beta}ds \right| + \left| \int_{\R \setminus [-M,M] } \frac{\mathbb{S}(\alpha,s)}{s-\beta}ds \right|.
\nonumber
\end{split}
\end{equation}
To prove the Lemma, it is enough to show that
\begin{equation} \label{0.3}
\begin{split}
&\limsup_{n\rightarrow \infty} \left |  \frac{1}{\widetilde{K}_n(0,0)} \int \frac{ K_n \left(\frac{\alpha}{\widetilde{K}_n(0,0)},t \right)}{t-\frac{\beta}{\widetilde{K}_n(0,0)}}d\mu(t)
-\int_{-M}^M \frac{\mathbb{S}(\alpha,s)}{s-\beta}ds \right|
\\
&\leq 8 \sqrt{\frac{ \|w\|_J}{w(0) M}}.
\end{split}
\end{equation}
Let $I_n=[-\frac{M}{\widetilde{K}_n(0,0)},\frac{M}{\widetilde{K}_n(0,0)}]$, and write
\begin{equation}
\begin{split}
& \left| \frac{1}{\widetilde{K}_n(0,0)} \int \frac{ K_n \left(\frac{\alpha}{\widetilde{K}_n(0,0)},t \right)}{t-\frac{\beta}{\widetilde{K}_n(0,0)}}d\mu(t)-\int_{-M}^M \frac{\mathbb{S}(\alpha,s)}{s-\beta}ds \right| \\
&\leq \left| \frac{1}{\widetilde{K}_n(0,0)} \int_{I_n} \frac{ K_n \left(\frac{\alpha}{\widetilde{K}_n(0,0)},t \right)}{t-\frac{\beta}{\widetilde{K}_n(0,0)}}d\mu(t) - \int_{-M}^M \frac{\mathbb{S}(\alpha,s)}{s-\beta}ds \right| \\
&+\left|\frac{1}{\widetilde{K}_n(0,0)} \int_{\R \setminus I_n} \frac{ K_n \left(\frac{a}{\widetilde{K}_n(0,0)},t \right)}{t-\frac{\beta}{\widetilde{K}_n(0,0)}}d\mu(t) \right|.
\end{split}
\nonumber
\end{equation}
Consider first the first term in the righthand side of the inequality above,
\begin{equation}
 \left| \frac{1}{\widetilde{K}_n(0,0)} \int_{I_n} \frac{ K_n \left(\frac{\alpha}{\widetilde{K}_n(0,0)},t \right)}{t-\frac{\beta}{\widetilde{K}_n(0,0)}}d\mu(t) - \int_{-M}^M \frac{\mathbb{S}(\alpha,s)}{s-\beta}ds \right|.
 \nonumber
\end{equation}
Since $0$ is not a pure point of $\mu$, it follows that $K_n(0,0) \rightarrow \infty$. Thus, for sufficiently large $n$, $I_n \subseteq J$. We can therefore write
\begin{equation}
\begin{split}
&\left| \frac{1}{\widetilde{K}_n(0,0)} \int_{I_n} \frac{ K_n \left(\frac{\alpha}{\widetilde{K}_n(0,0)},t \right)}{t-\frac{\beta}{\widetilde{K}_n(0,0)}}d\mu(t) - \int_{-M}^M \frac{\mathbb{S}(\alpha,s)}{s-\beta}ds \right|  \\
&= \left| \frac{1}{\widetilde{K}_n(0,0)} \int_{I_n} \frac{ K_n \left(\frac{a}{\widetilde{K}_n(0,0)},t \right)w(t)dt}{t-\frac{\beta}{\widetilde{K}_n(0,0)}} - \int_{-M}^M \frac{\mathbb{S}(\alpha,s)}{s-\beta}ds \right|,
\end{split}
\nonumber
\end{equation}
and by transformation of variables $t=\frac{s}{\widetilde{K}_n(0,0)}$ we get
\begin{equation}
\begin{split}
& \left| \frac{1}{\widetilde{K}_n(0,0)} \int_{I_n} \frac{ K_n \left(\frac{\alpha}{\widetilde{K}_n(0,0)},t \right)w(t)dt}{t-\frac{\beta}{\widetilde{K}_n(0,0)}} - \int_{-M}^M \frac{\mathbb{S}(\alpha,s)}{s-\beta}ds  \right| \\
&=\left| \int_{-M}^M \frac{1}{s-\beta} \left( \frac{K_n \left(\frac{\alpha}{\widetilde{K}_n(0,0)},\frac{s}{\widetilde{K}_n(0,0)} \right)w
\left( \frac{s}{\widetilde{K}_n(0,0)} \right)}{ K_n(0,0) w(0)}-\mathbb{S}(\alpha,s) \right) ds \right|.
\end{split}
\nonumber
\end{equation}
The righthand side of the equality  above goes to zero as $n \rightarrow \infty$. This follows since $K_n$ has a unifrom universal limit at $0$, since $w$ is essentially bounded on $J$ and since $0$ is a Lebesgue point of $w$ (note that $|s-\beta| \geq |\im(\beta)|>0$).

It remains to show that the inequality
\begin{equation} \label{final}
\limsup_{n \rightarrow \infty} \left|\frac{1}{\widetilde{K}_n(0,0)} \int_{\R \setminus I_n} \frac{ K_n \left(\frac{\alpha}{\widetilde{K}_n(0,0)},t \right)}{t-\frac{\beta}{\widetilde{K}_n(0,0)}}d\mu(t) \right|
\leq  8 \sqrt{\frac{ \|w\|_J}{w(0) M}}.
\end{equation}
holds.
First, by the Cauchy-Schwarz inequality we obtain
\begin{equation}
\begin{split}
&\left|\frac{1}{\widetilde{K}_n(0,0)} \int_{\R \setminus I_n} \frac{ K_n \left(\frac{\alpha}{\widetilde{K}_n(0,0)},t \right)}{t-\frac{\beta}{\widetilde{K}_n(0,0)}}d\mu(t) \right| \\
& \leq \left| \frac{1}{\widetilde{K}_n(0,0)}\right| \left( \int_{\R \setminus I_n} \frac{d \mu(t)}{\left(t-\frac{\beta}{\widetilde{K}_n(0,0)}\right)^2}   \right)^{1/2}\left( \int_{\R}K_n
 \left(\frac{\alpha}{\widetilde{K}_n(0,0)},t \right)^2d\mu(t)   \right)^{1/2}.
 \end{split}
\nonumber
\end{equation}
Second, the above inequality, and the reproducing property of the kernel imply
\begin{equation}
\begin{split}
&\left|\frac{1}{\widetilde{K}_n(0,0)} \int_{\R \setminus I_n} \frac{ K_n \left(\frac{\alpha}{\widetilde{K}_n(0,0)},t \right)}{t-\frac{\beta}{\widetilde{K}_n(0,0)}}d\mu(t) \right|\\
&\leq\left| \frac{K_n \left(\frac{\alpha}{\widetilde{K}_n(0,0)},\frac{\alpha}{\widetilde{K}_n(0,0)} \right)^{1/2}}{w(0) K_n(0,0)}\right| \left( \int_{\R \setminus I_n} \frac{d \mu(t)}{\left( t-\frac{\beta}{\widetilde{K}_n(0,0)}\right)^2}   \right)^{1/2}.
\end{split}
\nonumber
\end{equation}
Note that for sufficiently large $n$,
\begin{equation}
\left|\frac{K_n \left(\frac{\alpha}{\widetilde{K}_n(0,0)},\frac{\alpha}{\widetilde{K}_n(0,0)} \right)^{1/2}}{w(0) K_n(0,0)}\right|
\leq \frac{2}{w(0) K_n(0,0)^{1/2}},
\nonumber
\end{equation}
by assumption $(iii)$ in the statement of Theorem \ref{MainTheorem}. Thus we get, for sufficiently large $n$,
\begin{equation} \label{4}
\begin{split}
&\left|\frac{1}{\widetilde{K}_n(0,0)} \int_{\R \setminus I_n} \frac{ K_n \left(\frac{\alpha}{\widetilde{K}_n(0,0)},t \right)}{t-\frac{\beta}{\widetilde{K}_n(0,0)}}d\mu(t) \right|  \\
& \leq \frac{2}{w(0) K_n(0,0)^{1/2}} \left( \int_{\R \setminus I_n} \frac{d \mu(t)}{\left(t-\frac{\beta}{\widetilde{K}_n(0,0)}\right)^2}\right)^{1/2}.
\end{split}
\end{equation}
We are left with estimating
\begin{equation}
\int_{\R \setminus I_n} \frac{d \mu(t)}{\left(t-\frac{\beta}{\widetilde{K}_n(0,0)}\right)^2}.
\nonumber
\end{equation}
First note that, since $M \geq 2|\beta|$, we get for every $t$ with $|t| \geq \frac{M}{\widetilde{K}_n(0,0)}$,
$$
|t-\frac{\beta}{\widetilde{K}_n(0,0)}|\geq \frac{t}{2}.
$$
Therefore,
\begin{equation}
\int_{\R \setminus I_n} \frac{d \mu(t)}{\left(t-\frac{\beta}{\widetilde{K}_n(0,0)}\right)^2} \leq \int_{\R \setminus I_n} \frac{4 d \mu(t)}{t^2}.
\nonumber
\end{equation}
Now we split the integral in the righthand side of the inequality above. Let $H_n=[-\frac{M}{\widetilde{K}_n(0,0)^{1/3}},\frac{M}{\widetilde{K}_n(0,0)^{1/3}}]$, and again note that for sufficiently large $n$, $H_n \subseteq J$. Clearly, $I_n \subseteq H_n$, so we can write
$$
\R \setminus I_n= \left( \R \setminus H_n \right) \cup \left(H_n \setminus I_n \right).
$$
We split the integral accordingly,
\begin{equation}
\begin{split}
\left| \int_{\R \setminus I_n} \frac{4 d \mu(t)}{t^2} \right| & = \left| \int_{\R \setminus H_n} \frac{4 d \mu(t)}{t^2}+\int_{H_n \setminus I_n} \frac{4 d \mu(t)}{t^2}\right| \\
& =\left| \int_{\R \setminus H_n} \frac{4 d \mu(t)}{t^2}+\int_{H_n \setminus I_n} \frac{4 w(t)dt}{t^2}\right| \\
& \leq \left|  \int_{\R \setminus H_n} \frac{4 d \mu(t)}{t^2} \right| + 4\|w\|_J \cdot \left|  \int_{H_n\setminus I_n}\frac{dt}{t^2} \right|,
\end{split}
\nonumber
\end{equation}
where we have used the fact that $H_n \subseteq J$ to write $d\mu(t)=w(t)dt$ there.
The proof of Lemma \ref{1} is almost complete. For the first integral on the left, note that for $t \notin H_n$, $t^2 \geq \frac{M^2}{\widetilde{K}_n(0,0)^{2/3}}$. Therefore, we have (taking into account that $\mu$ is a probability measure)
\begin{equation}
 \left|  \int_{\R \setminus H_n} \frac{4 d \mu(t)}{t^2} \right| \leq 4 \frac{\widetilde{K}_n(0,0)^{2/3}}{M^2}\int d\mu(t)
= 4 \frac{\widetilde{K}_n(0,0)^{2/3}}{M^2}.
\nonumber
\end{equation}
For the second integral, integration of $\frac{1}{t^2}$ gives (recall the definition of $I_n$)
\begin{equation}
4\|w\|_J \cdot \left|  \int_{H_n\setminus I_n}\frac{dt}{t^2} \right| \leq 4 \|w\|_J \cdot \int_{\R \setminus I_n}\frac{dt}{t^2}=8 \|w\|_J \frac{\widetilde{K}_n(0,0)}{M}.
\nonumber
\end{equation}
Plugging these estimates into (\ref{4}), we see
\begin{equation}
\begin{split}
&\left|\frac{1}{\widetilde{K}_n(0,0)} \int_{\R \setminus I_n} \frac{ K_n \left(\frac{\alpha}{\widetilde{K}_n(0,0)},t\right)}{t-\frac{\beta}{\widetilde{K}_n(0,0)}}d\mu(t) \right|  \\
& \leq \frac{8}{w(0) K_n(0,0)^{1/2}} \left( \frac{\widetilde{K}_n(0,0)^{2/3}}{M^2}+\|w\|_J  \frac{\widetilde{K}_n(0,0)}{M}  \right)^{1/2} \\
&= \frac{8}{\sqrt{w(0)}} \left( \frac{1}{M^2 \widetilde{K}_n(0,0)^{1/3}}+\|w\|_J  \frac{1}{M}  \right)^{1/2}
\end{split}
\end{equation}
which immediately shows that
\begin{equation}
\limsup_{n \rightarrow \infty} \left|\frac{1}{\widetilde{K}_n(0,0)} \int_{\R \setminus I_n} \frac{ K_n \left(\frac{\alpha}{\widetilde{K}_n(0,0)},t \right)}{t-\frac{\beta}{\widetilde{K}_n(0,0)}}d\mu(t) \right|
\leq 8 \sqrt{\frac{ \|w\|_J}{w(0) M}}.
\nonumber
\end{equation}
\end{proof}

\end{document}